\documentclass{ws-procs975x65}

\def\3nab{\tilde{\nabla}}

\def\la {\langle}
\def\ra {\rangle}

\def\be {\begin{equation}}
\def\ee {\end{equation}}
\def\ba {\begin{eqnarray}}
\def\ea {\end{eqnarray}}
\newcommand{\bra}[1]{\left(#1\right)}

\newcommand{\brac}[1]{\left\{#1\right\}}

\newcommand{\sfr}[2]{{\textstyle\frac{#1}{#2}}}

\newcommand{\E}{{\mathcal E}}
\renewcommand{\H}{{\mathcal H}}
\newcommand{\barray}{\begin{array}}
\newcommand{\earray}{\end{array}}

\newcommand \ep {\epsilon}

\newcommand \om {\omega}

\newcommand{\udot}{{\mathcal A}}

\begin{document}
\title{Notes on Cosmic Censorship Conjecture revisited: Covariantly}
\author{Hamid, Aymen I. M.$^*$, Goswami, Rituparno$^\dagger$ and Maharaj, Sunil D.$^{\dagger\dagger}$}

\address{Astrophysics and Cosmology Research Unit, School of Mathematics, Statistics and Computer Science, University of KwaZulu-Natal, Private Bag X54001, Durban 4000, South Africa.\\
$^*$E-mail: aymanimh@gmail.com\footnote{Physics Department, Faculty of Science, University of Khartoum, Sudan}\\$^{\dagger}$E-mail:Goswami@ukzn.ac.za\\$^{\dagger\dagger}$E-mail: Maharaj@ukzn.ac.za}

\begin{abstract}
In this paper we study the dynamics of the trapped region using a frame independent semi-tetrad covariant formalism for general Locally Rotationally Symmetric (LRS) class II spacetimes. We covariantly prove some important geometrical results for the apparent horizon, and state the necessary and sufficient conditions for a singularity to be locally naked. These conditions bring out, for the first time in a quantitative and transparent manner, the importance of the Weyl curvature in deforming and delaying the trapped region during continual gravitational collapse, making the central singularity locally visible.
\end{abstract}


\bodymatter


\section{Introduction}

Since Penrose proposed the famous {\it  Cosmic Censorship Conjecture (CCC)} in 1969 \cite{CCC}, stating that 
singularities observable from the outside will never arise in generic gravitational collapse which starts from a perfectly reasonable 
nonsingular initial state, there have been numerous attempts towards validating this conjecture  by means of a rigorous mathematical treatment. 
However, this conjecture remains unproved, and it has been recognised as one of the most important open problems in gravitational physics. 
The key point here is that the validity of this conjecture will confirm the already widely accepted 
and applied theory of black hole dynamics, which has considerable amount of astrophysical applications. 
On the other hand, it's overturn will throw the black hole dynamics into serious doubt. This is because most of the important fundamental global theorems 
in black hole physics assume that the spacetime manifold is {\it future asymptotically predictable}. In other words this condition ensures that there should be 
no singularity to the future of the partial Cauchy surface which is `naked' or visible from future null infinity \cite{HE}.

Although no conclusive proof or disproof of CCC could be formulated, the quest gave rise to a number 
of counter examples which showed there are shell focusing naked singularities occurring at the centre of spherically symmetric dust, perfect fluids or radiation shells
(see for example Ref.~\refcite{Goswami:2006ph, Joshibook1} and references therein).
We can, in principle, rule out these naked singularities by stating that dust or perfect fluids are not really 
`fundamental' forms of matter field, as their properties are not derived from a `proper' Lagrangian. 
However, if cosmic censorship is to be established as a rigorous mathematical theorem, this objection 
has to be made precise in terms of a clear and simple restriction on the energy momentum tensor of the matter field. This paper discusses the CCC using the semi-tetrad covariant formalisms applied to LRS-II spacetimes.
\section{Semi-tetrad covariant formalisms}
In the 1+3 covariant formalism the covariant derivative of the time like vector $u^a$ can be
decomposed into the irreducible part as $\nabla_au_b=-A_au_b+\frac13h_{ab}\Theta+\sigma_{ab}+\ep_{a b
c}\om^c$, where $A_a=\dot{u}_a$ is the acceleration,
$\Theta=D_au^a$ is the expansion, $\sigma_{ab}=D_{\la a}u_{b \ra}$
is the shear tensor and $w^a=\ep^{a b c}D_bu_c$ is the vorticity
vector. 
The energy momentum tensor for a general matter field can 
be similarly decomposed as follows
\be
T_{ab}=\mu u_au_b+q_au_b+q_bu_a+ph_{ab}+\pi_{ab}\;,
\ee
where $\mu=T_{ab}u^au^b$ is the energy density, $p=(1/3 )h^{ab}T_{ab}$ is the isotropic pressure, $q_a=q_{\la a\ra}=-h^{c}{}_aT_{cd}u^d$ is the 3-vector defining 
the heat flux and $\pi_{ab}=\pi_{\la ab\ra}$ is the anisotropic stress.

We choose to work on the isotropic LRS spacetimes where all vector and tensor vanishes, and the non-zero 1+1+2 variables are covariantly defined scalars are $\brac{\udot, \Theta,\phi, \xi, \Sigma,\Omega, \E, \H, \mu, p, \Pi, Q }$. 
Within the LRS cases is the LRS-II class that admits spherically symmetric solutions and is free of rotation, thus allowing for the vanishing of the variables $\Omega$, $ \xi $ and $ \H $. The set of quantities that fully describe LRS class II spacetimes are $\brac{\udot, \Theta,\phi, \Sigma,\E, \mu, p, \Pi, Q }$. Using the field equations, we can write the evolution and propagation equations for the Gaussian curvature $K$ are $\dot K = -\bra{\sfr23\Theta-\Sigma}K\ , \hat K = -\phi K$.

\section{Apparent Horizon in spherically symmetric spacetimes}
At this point let us define the notion of locally {\it outgoing} and {\it incoming} null geodesics with respect to the preferred spatial direction.
For LRS-II spacetimes the outgoing null vector is defined as $k^a=\frac{E}{\sqrt{2}}\bra{u^a+e^a}$.
Together with the covariant derivative given by \cite{deSwardt:2010nf}
\be
\nabla_bk_a=\sfr12 \tilde{h}_{ab}\tilde{\Theta}_{out}+\tilde{\sigma}_{ab}+\tilde{\omega}_{ab}+\tilde{X}_ak_b+\tilde{Y}_bk_a
+\lambda k_ak_b,
\ee
where $\tilde{X}_a$, $\tilde{Y}_a$ and $\lambda$ are combinations in the above covariant derivative with $e^a$ 
$\tilde{h}_{ab}=N_{ab}$ equivalent to the projection tensor on the 2-sphere, and  $\tilde{\Theta}_{out}$, $\tilde{\sigma}_{ab}$, $\tilde{\omega}_{ab}$ represents the expansion, shear and vorticity 
of the outgoing null congruence respectively. A similar decomposition can be done for the incoming null geodesic $l^a$. In order to define the trapped region of collapsing star, we follow the definition of the boundary of the trapped region which defines the {\it apparent horizon} as in \cite{HE}, and that surface is attained when the volume expansion of the outgoing null congruence vanishes ($\tilde{\Theta}_{out}=0$). 
\begin{proposition}
For any spherically symmetric spacetime $(\mathcal{M},g)$ that allows a local 1+1+2 splitting, the apparent horizon is described by the 
curve $ \bra{\sfr23\Theta-\Sigma+\phi}=0 $, while the cosmological horizon is described by $ \bra{\sfr23\Theta-\Sigma-\phi}=0 $, in the local $[u,e]$ plane.
\end{proposition}
\begin{proof}
Working out $\tilde{\Theta}_{out}$ we conclude that $\tilde{\Theta}_{out}=\tilde{h}^{ab}\nabla_bk_a=E\bra{\sfr23\Theta-\Sigma+\phi}$.
Hence for a null congruence with non-zero energy $E$, $\tilde{\Theta}_{out}=0$ implies that $ \bra{\sfr23\Theta-\Sigma+\phi}=0 $.
$\tilde{\Theta}_{in}=0$ will then imply $ \bra{\sfr23\Theta-\Sigma-\phi}=0 $.
 \end{proof}
\begin{proposition}
For any spherically symmetric spacetime $(\mathcal{M},g)$ that allows a local 1+1+2 splitting, the gradient of the Gaussian curvature of 
the 2-sheets that intersect with the apparent (or cosmological) horizon is null.
\end{proposition}
\begin{proof}
We calculate the quantity $\nabla_a K\nabla^aK$ for a spherically symmetric spacetime (where $K\ne 0$):
\be
\nabla_a K\nabla^aK=\bra{\sfr23\Theta-\Sigma+\phi}\bra{\sfr23\Theta-\Sigma-\phi}K^2\label{Geq}\;.
\ee
Hence for the 2-sheets intersecting the horizon (apparent or cosmological), the gradient of their Gaussian curvature is null.
\end{proof}
As we are considering the scenario of gravitational collapse of massive stars, henceforth we will only concentrate on the apparent horizon. We have already seen that the 
curve $\Psi\equiv\frac23\Theta-\Sigma+\phi=0$,
describes the apparent horizon. Let the vector $\Psi^a=\alpha u^a+\beta e^a$ be the tangent to the curve in the local $[u,e]$ plane. Then it is easy to obtain by the field equations
\be\label{alph}
\frac{\alpha}{\beta}=\frac{\sfr23\mu+\sfr12\Pi+{\cal E}-Q}{-\sfr13\mu-p+{\cal E}-\sfr12 \Pi+Q}\;.
\ee
It is interesting to note that the thermodynamical quantities together with the Weyl scalar completely determine the tangent to the apparent horizon. 
We will define the apparent horizon to be locally {\it outgoing} at a point $p\in[u,e]$, if the slope of the tangent to the horizon is positive at $p$, that is 
$\sfr{\alpha}{\beta}>0$. Let the point $p$ be labelled by the values of the local coordinates ($t_0,\chi_0$) which are the affine parameters along the 
integral curves of $u^a$ and $e^a$ 
respectively. Then a locally outgoing apparent horizon at $p$ would imply that the 2-sheets (spherical shell) labelled by $\chi_0+\epsilon$ will get trapped 
later than $t=t_0$, while the 2-sheet labelled by $\chi_0$ gets trapped at $t=t_0$. Finally as we can easily see that the sign of the scalar $\Psi^a\Psi_a$, determines whether the 
curve $\Psi=0$ is timelike, spacelike or null in the $[u,e]$ plane. Hence $\sfr{\alpha^2}{\beta^2}>(<)1$ denotes the horizon to be locally timelike (spacelike). If $\sfr{\alpha^2}{\beta^2}=1$ then the horizon is null.
We now would like to drive the necessary and sufficient condition for existence of a locally naked singularity. 

\begin{proposition}
Consider the continued collapse of a general spherically symmetric matter cloud from a regular initial epoch which obeys the physically reasonable energy conditions.
If the following conditions are satisfied:
\begin{enumerate}
\item The spacetime is free of shell crossing singularities,
\item Closed trapped surfaces exist,
\end{enumerate}
then the necessary and sufficient condition for the central singularity to be locally naked is that the slope of the tangent to the apparent horizon at the central singularity is
positive and non spacelike ($\sfr{\alpha}{\beta}\ge1$).
\end{proposition}
\begin{proof} 
Let the central singularity be denoted by ($t=t_{s_0}$, $\chi=0$) in the $[u,e]$ plane. The key point here is that there should be available an untrapped region in the local neighbourhood 
of the central singularity for a null geodesic with the past end point arbitrarily near the central singularity to escape. We have assumed here that the singularity curve is a non-decreasing  function of the affine parameter of the integral curve of the vector $e^a$ (see \cite{Goswami:2002he}), and hence no other collapsing shell becomes singular before the central shell. If the apparent horizon at the central singularity is ``{\em ingoing}", that is $\sfr{\alpha}{\beta}<0$, then the neighbourhood of the centre gets trapped before the central singularity and no null geodesic from a point arbitrarily close to the central singularity can escape. Also 
if the apparent horizon is ``{\em outgoing}" but spacelike, that is $0\le\sfr{\alpha}{\beta}<1$, then any outgoing null direction from the central singularity will be necessarily within the trapped region. Hence for these cases, any null geodesic from a point arbitrarily close to the central singularity will have $\tilde\Theta_{out}<0$ and hence they will fall to the 
singularity. Therefore the necessary condition for a singularity to be  locally naked is that the slope of the tangent to the apparent horizon at the central singularity is
positive and non-spacelike ($\sfr{\alpha}{\beta}\ge1$).
Conversely, suppose there exist a family of future directed null geodesics that has escaped from the points arbitrarily close to the central singularity in the $[u,e]$ plane. Then that would imply these points are non-trapped and the slope of the apparent horizon curve at the central singularity is greater than (or equal to) the slope of these outgoing null geodesics in order for them to escape. Hence $\sfr{\alpha}{\beta}\ge1$ is the necessary and sufficient condition for the singularity to be locally naked.
\end{proof} 
\begin{proposition}
Consider the gravitational collapse of a spherically symmetric perfect fluid obeying strong energy conditions $\mu\ge0$ and $\mu+3p\ge0$.
If the following conditions are satisfied:
\begin{enumerate}
\item The spacetime is free of shell crossing singularities,
\item Closed trapped surfaces exist,
\item The central singularity is marginally naked ($\sfr{\alpha}{\beta}=1$),
\end{enumerate}
then the limit of $\sfr{|{\cal E}|}{\mu+p}$ at the central singularity along the apparent horizon curve diverges.
\end{proposition}
\begin{proof}
We know that for a perfect fluid we have $Q=\Pi=0$, and at the central singularity $\sfr{\alpha}{\beta}=1$ implies
\be
\left[\frac{{\cal E}}{\mu+p}-\frac13\frac{\mu+3p}{\mu+p}\right]^{-1}=0.
\ee
For the perfect fluid satisfying the strong energy conditions, $\sfr{\mu+3p}{\mu+p}$ is finite and hence $\sfr{|{\cal E}|}{\mu+p}$ at the central singularity 
along the apparent horizon tends to infinity.
\end{proof}
The above result clearly shows that the electric part of the Weyl scalar (which is responsible for the tidal forces) must diverge faster than the energy density along the 
apparent horizon curve, for a singularity to be locally naked. It also highlights the importance of tidal forces in delaying the trapping. In fact, this results closely relates to the result obtained in Ref.~\refcite{Joshi:2001xi}. 
\begin{corollary}
Consider the continued gravitational collapse of a spherically symmetric perfect fluid obeying the strong energy condition $\mu\ge0$ and $\mu+3p\ge0$. If the 
spacetime is conformally flat then the end state of the collapse
is necessarily a black hole.
\end{corollary}
\begin{proof}
Conformally flat spacetime implies vanishing of the Weyl tensor. Hence we have ${\cal E}=0$. Also for a perfect fluid $Q=\Pi=0$. We therefore have $\frac{\alpha}{\beta}=\frac{-\sfr23\mu}{\sfr13\mu+p}$.
Now the condition $\sfr{\alpha}{\beta}\ge1$ implies $\mu+p\le0$ which violates the strong energy condition.
In fact we can explicitly calculate the norm of the tangent to show that $\Psi^a\Psi_a \propto -\frac{1}{3}(\mu+p)(\mu-3p)$.
If the strong energy condition is satisfied  with $\mu+p>0$, then we have the following cases: If $\mu>3p$ then the $\Psi^a$ is ``ingoing'' timelike, whereas if $\mu=3p$ then the $\Psi^a$ is ``ingoing'' null, and finally if $\mu<3p$ then the $\Psi^a$ is ``ingoing'' spacelike.
In all these cases the region around the centre gets trapped before the central singularity. Hence the singularity is always covered and the collapse end state 
is always a black hole.
\end{proof}
Some examples that can be applied to the above considerations are found in Ref.~\refcite{Hamid:2014kza}

\section{Discussion}

In this paper, working in a covariant  and frame independent formalisms, we successfully identified the physical and geometrical mechanisms responsible 
for delaying the trapped surface formation and making the central singularity locally naked during the continued gravitational collapse of a massive star. 
By working out the dynamics of the trapped region we transparently and quantitatively identified the role of Weyl curvature in deforming the trapped region in such a 
way that the singularity can be naked. As we know the Weyl curvature is responsible for the tidal force between nearby geodesics that generates the 
spacetime shear. In fact from the field equations of LRS-II spacetimes we can immediately see that the Weyl scalar is the source term for the 
shear evolution equation. Spacetime shear then deforms the apparent horizon and delays the trapping as shown in  Ref.~\refcite{Joshi:2001xi, Joshi:2004tb}.


\section*{Acknowledgments}

AH would like to thank Radouane Gannouji for the useful discussions. AH and RG are supported by National Research Foundation (NRF), South Africa. SDM 
acknowledges that this work is based on research supported by the South African Research Chair Initiative of the Department of
Science and Technology and the National Research Foundation.

\end{document}